\documentclass[12pt]{amsart}

\usepackage[centertags]{amsmath}
\usepackage{amsthm,amsfonts}
\usepackage{xcolor}

\usepackage{amssymb}
\usepackage{epsfig}
\usepackage{amssymb,latexsym}

\begin{document}

\theoremstyle{plain}
\newtheorem{theorem}{Theorem}[section]
\newtheorem{lemma}[theorem]{Lemma}
\newtheorem{corollary}[theorem]{Corollary}
\newtheorem{proposition}[theorem]{Proposition}
\newtheorem{question}[theorem]{Question}
\theoremstyle{definition}
\newtheorem{prop}{Proposition}[theorem]
\newtheorem{obs}{Observation}[theorem]
\newtheorem{remark}[theorem]{Remark}
\newtheorem{cor}{Corollary}[theorem]
\newtheorem{example}[theorem]{Example} 
\newtheorem{definition}[theorem]{Definition}

\newcommand{\binomial}[2]{\left(\begin{array}{c}#1\\#2\end{array}\right)}
\newcommand{\zar}{{\rm zar}}
\newcommand{\an}{{\rm an}}
\newcommand{\red}{{\rm red}}
\newcommand{\codim}{{\rm codim}}
\newcommand{\rank}{{\rm rank}}
\newcommand{\Pic}{{\rm Pic}}
\newcommand{\Div}{{\rm Div}}
\newcommand{\Hom}{{\rm Hom}}
\newcommand{\im}{{\rm im}}
\newcommand{\Spec}{{\rm Spec}}
\newcommand{\sing}{{\rm sing}}
\newcommand{\reg}{{\rm reg}}
\newcommand{\Char}{{\rm char}}
\newcommand{\Tr}{{\rm Tr}}
\newcommand{\res}{{\rm res}}
\newcommand{\tr}{{\rm tr}}
\newcommand{\supp}{{\rm supp}}
\newcommand{\Gal}{{\rm Gal}}
\newcommand{\Min}{{\rm Min \ }}
\newcommand{\Max}{{\rm Max \ }}
\newcommand{\Span}{{\rm Span  }}

\newcommand{\Frob}{{\rm Frob}}
\newcommand{\lcm}{{\rm lcm}}


\long\def\symbolfootnote[#1]#2{\begingroup%
\def\thefootnote{\fnsymbol{footnote}}\footnote[#1]{#2}\endgroup}

\newcommand{\soplus}[1]{\stackrel{#1}{\oplus}}
\newcommand{\dlog}{{\rm dlog}\,}    
\newcommand{\limdir}[1]{{\displaystyle{\mathop{\rm
lim}_{\buildrel\longrightarrow\over{#1}}}}\,}
\newcommand{\liminv}[1]{{\displaystyle{\mathop{\rm
lim}_{\buildrel\longleftarrow\over{#1}}}}\,}
\newcommand{\boxtensor}{{\Box\kern-9.03pt\raise1.42pt\hbox{$\times$}}}
\newcommand{\sext}{\mbox{${\mathcal E}xt\,$}}
\newcommand{\shom}{\mbox{${\mathcal H}om\,$}}
\newcommand{\coker}{{\rm coker}\,}
\renewcommand{\iff}{\mbox{ $\Longleftrightarrow$ }}
\newcommand{\onto}{\mbox{$\,\>>>\hspace{-.5cm}\to\hspace{.15cm}$}}

\newenvironment{pf}{\noindent\textbf{Proof.}\quad}{\hfill{$\Box$}}

\newcommand{\sA}{{\mathcal A}}
\newcommand{\sB}{{\mathcal B}}
\newcommand{\sC}{{\mathcal C}}
\newcommand{\sD}{{\mathcal D}}
\newcommand{\sE}{{\mathcal E}}
\newcommand{\sF}{{\mathcal F}}
\newcommand{\sG}{{\mathcal G}}
\newcommand{\sH}{{\mathcal H}}
\newcommand{\sI}{{\mathcal I}}
\newcommand{\sJ}{{\mathcal J}}
\newcommand{\sK}{{\mathcal K}}
\newcommand{\sL}{{\mathcal L}}
\newcommand{\sM}{{\mathcal M}}
\newcommand{\sN}{{\mathcal N}}
\newcommand{\sO}{{\mathcal O}}
\newcommand{\sP}{{\mathcal P}}
\newcommand{\sQ}{{\mathcal Q}}
\newcommand{\sR}{{\mathcal R}}
\newcommand{\sS}{{\mathcal S}}
\newcommand{\sT}{{\mathcal T}}
\newcommand{\sU}{{\mathcal U}}
\newcommand{\sV}{{\mathcal V}}
\newcommand{\sW}{{\mathcal W}}
\newcommand{\sX}{{\mathcal X}}
\newcommand{\sY}{{\mathcal Y}}
\newcommand{\sZ}{{\mathcal Z}}

\newcommand{\A}{{\mathbb A}}
\newcommand{\B}{{\mathbb B}}
\newcommand{\C}{{\mathbb C}}
\newcommand{\D}{{\mathbb D}}
\newcommand{\E}{{\mathbb E}}
\newcommand{\F}{{\mathbb F}}
\newcommand{\G}{{\mathbb G}}
\newcommand{\HH}{{\mathbb H}}
\newcommand{\I}{{\mathbb I}}
\newcommand{\J}{{\mathbb J}}
\newcommand{\M}{{\mathbb M}}
\newcommand{\N}{{\mathbb N}}
\renewcommand{\P}{{\mathbb P}}
\newcommand{\Q}{{\mathbb Q}}
\newcommand{\T}{{\mathbb T}}
\newcommand{\U}{{\mathbb U}}
\newcommand{\V}{{\mathbb V}}
\newcommand{\W}{{\mathbb W}}
\newcommand{\X}{{\mathbb X}}
\newcommand{\Y}{{\mathbb Y}}
\newcommand{\Z}{{\mathbb Z}}


\newcommand{\Fqm}{\mathbb{F}_{q^m}}
\newcommand{\Fq}{\mathbb{F}_q}
\newcommand{\Fp}{\mathbb{F}_p}
\newcommand{\Fpl}{\mathbb{F}_{p^l}}
\newcommand{\fqn}{\mathbb{F}_q^n}
\newcommand{\be}{\begin{eqnarray}}
\newcommand{\ee}{\end{eqnarray}}
\newcommand{\nn}{{\nonumber}}
\newcommand{\dd}{\displaystyle}
\newcommand{\ra}{\rightarrow}
\newcommand{\bigmid}[1][12]{\mathrel{\left| \rule{0pt}{#1pt}\right.}}
\newcommand{\cl}{${\rm \ell}$}
\newcommand{\clp}{${\rm \ell^\prime}$}
\title[]{On the Hull and Complementarity of One Generator Quasi-Cyclic Codes and Four-Circulant Codes}
\author[]{Zohreh Aliabadi, Cem G\"uner\.{I}, Tekg\"ul Kalayc\i \\
}

\maketitle
\begin{center}
Sabanc\i~ University, Faculty of Engineering and Natural Sciences, 34956 \.{I}stanbul, Turkey\\
E-mail: \{zaliabadi, cem.guneri, tekgulkalayci\}@sabanciuniv.edu\\
\end{center}
\begin{abstract}
We study one generator quasi-cyclic codes and four-circulant codes, which are also quasi-cyclic but have two generators. We state the hull dimensions for both classes of codes in terms of the polynomials in their generating elements. We prove results such as the hull dimension of a four-circulant code is even and one dimensional hull for double-circulant codes, which are special one generator codes, is not possible when the alphabet size $q$ is congruent to 3 mod 4. We also characterize linear complementary pairs among both classes of codes. Computational results on the code families in consideration are provided as well.  
\end{abstract}


\section{Introduction}
\label{sec1}
The hull of a linear code $C$ is defined as $C\cap C^{\perp}$, where $C^{\perp}$ is the Euclidean dual code. This concept was introduced by Assmus and Key in \hspace{1sp}\cite{AFFINE} and later found applications in various problems, such as determining permutation equivalence between codes (\cite{PERBET}) and construction of quantum error-correcting codes (\cite{1DIMAG}).   The hull of a linear code can also be defined with respect to other inner products (see  \cite{GALOISHULL} and \cite{GALOISLIN} for the study of the hull with respect to the Galois inner product).

The smallest possible hull dimension is 0 and codes having trivial hull are called linear complementary dual (LCD) codes. LCD codes were introduced by Massey in \cite{LCD}. Note that the name LCD is justified, since $C\oplus C^{\perp}=\mathbb{F}_q^n$ for an LCD code $C \subseteq \mathbb{F}_q^n$. LCD codes are generalized to linear complementary pair (LCP) of codes, where a pair $(C,D)$ of linear codes in $\F_q^n$ is called LCP of codes if $C\oplus D=\F_q^n$. LCD and LCP of codes have drawn much attention in recent years due to their applications in cryptography in the context of side channel and fault injection attacks (see \cite{LCDSTRE}, \cite{SCA}, \cite{SCAFIA}). In this application, the security parameter of an LCP $(C, D)$ of codes is defined as $\min \lbrace d(C), d(D^{\perp})\rbrace$. LCD and LCP of codes have been very actively studied and we refer to \cite{LCP}, \cite{LINEQLCD}, \cite{LCDQC}, \cite{4CLCD} for some of the recent developments. 

The next smallest hull dimension is 1, which is also of interest. We refer to some of the recent contributions in the literature, where codes with hull dimension 1 are studied particularly (\hspace{1sp}\cite{1DIMSEMI}, \cite{1dim}, \cite{1DIMAG}).

Quasi-cylic (QC) codes is one of the well-studied families in coding theory. The general theory of QC codes is laid out in \cite{AQC, GENERATOR} and one generator QC codes are throgouhly studied in \cite{1GQC}. LCD and LCP of general QC codes are addressed in \cite{LCDQC} and in \cite{LCP}, respectively. This article investigates the hull and complementarity of special classes of QC codes to obtain concrete results in terms of the polynomials in their generating elements. We study one generator QC codes and four-circulant (FC) codes, which are also QC codes but with two generators. A particular case of one generator QC codes, that is double-circulant (DC) codes, is also addressed. Section 2 provides the background material needed on the hull dimension of linear codes and on QC codes.  In Section 3, we prove a formula for the hull dimension of one generator QC codes and also characterize the LCP of one generator QC codes. For the special case of DC codes, we provide a necessary and sufficient condition for one hull dimension. Section 4 studies FC codes. Complementary dual FC codes were studied in \cite{4CLCD}. Here, we state a formula for the hull dimension of FC codes, which shows that there exists no FC code with an odd hull dimension. We also characterize the LCP of FC codes. By the results of Carlet et al. (\cite{LINEQLCD}, \cite{SIGMALCD}), optimal parameter problem for LCD and LCP of codes are of interest for binary and ternary codes. Motivated by this, computational results on the parameters of binary and ternary LCD and LCP codes, as well as codes with one-dimensional hull, are presented for the code families studied in this paper. Some of the parameters we find are optimal according to \cite{CDTB}.

\section{Preliminaries}\label{prelim}
For a linear $[n,k]$ code $C$ over a finite field $\F_q$, the hull of $C$ is defined as 
$$\hbox{Hull}(C):=C\cap C^{\perp},$$
where $C^\bot$ denotes the dual of $C$ with respect to Euclidean inner product. Let $h(C)=\dim(\hbox{Hull}(C))$ denote the hull dimension. If $q$ is a square one can also define the Hermitian hull of $C$
 $$\hbox{Hull}_h(C):=C \cap C^{\perp_h},$$
where $C^{\bot_h}$ is the dual with respect to Hermitian inner product. We denote the Hermitian hull dimension by $h_h(C)=\dim(\hbox{Hull}_h(C))$.

If $G$ denotes the $k\times n$ generator matrix of $C$, we have 
\begin{equation}\label{hull dim}
h(C)=k-\hbox{rank}(GG^T) \ \mbox{(\cite[Proposition 3.1]{HULL=RANK})}.
\end{equation}
If we denote by $\bar{G}$ the matrix obtained from $G$ by raising all entries to power $\sqrt{q}$, then the Hermitian hull dimension is given by 
\begin{equation}\label{herm hull dim}
h_h(C)=k-\hbox{rank}(G\bar{G}^T) \ \mbox{(\cite[Proposition 3.2]{HULL=RANK})}.
\end{equation}

The minimum hull dimension is 0, which amounts to $C\cap C^\bot=\{0\}$ (equivalently $C\oplus C^\bot =\F_q^n$). In this case, $C$ is called a linear complementary dual (LCD) code. A pair $(C,D)$ of linear codes of length $n$ over $\mathbb{F}_q$ is called a linear complementary pair (LCP) of codes if $C\oplus D=\mathbb{F}_q^n$. It is clear that LCD codes is a special case of LCP of codes. Namely, $C$ is LCD if and only if $(C,C^\bot)$ is LCP of codes. The security parameter of LCP of codes $(C,D)$ is defined as $\min \lbrace d(C), d(D^{\perp})\rbrace$. Note that for an LCD code $C$, the security parameter is simply $d(C)$. 

An $[n,k]$ linear code $C$ over $\F_q$ is called a quasi-cyclic (QC) code of index $\ell$ if its codewords are invariant under shift by $\ell$ units, and $\ell$ is the smallest positive integer with this property. It is known that the index of a QC code is a divisor of its length, say $ n=m\ell$. Here, $m$ is referred to as the co-index of $C$.  We refer to a QC code of index $\ell$ as $\ell$-QC code for simplicity. 

It is clear that cyclic codes correspond to QC codes of index 1. Similar to cyclic codes, QC codes also have rich algebraic structures. Let us denote the space of $m\times \ell$ arrays over $\F_q$ by $\F_q^{m\times \ell}$ and view an $\ell$-QC code of length $m\ell$ as a subspace of this space. With this notation, $C$ being index $\ell$ amounts to codewords being closed under row shift. If we let $R_m:=\mathbb{F}_q[x]/\langle x^m-1\rangle$, then the following map induces a one-to-one correspondence between $\ell$-QC codes in $\F_q^{m\times \ell}$ and $R_m$-submodules of $R_m^\ell$ (\cite[Lemma 3.1]{AQC}): 
\begin{equation*}\begin{array}{cccc} \label{identification}
\phi: &  \F_q^{m\ell} & \longrightarrow & R_m^\ell  \\
& c=\left(c_{ij}\right) & \longmapsto & (c_0(x),c_1(x),\ldots ,c_{\ell-1}(x)) ,
\end{array}\end{equation*}
where 
$$c_j(x):=  \displaystyle{ \sum_{i=0}^{m-1}} c_{ij}x^i =  c_{0j}+c_{1j}x+c_{2j}x^2+\cdots + c_{m-1,j}x^{m-1} \in R_m $$
for each $0\leq j \leq \ell-1$. 

\textbf{Throughout the manuscript, we assume that $q$ and $m$ are relatively prime.}  With this assumption, we have the following factorization into distinct irreducible polynomials in $\F_q[x]$:
\begin{equation}\label{factors}
x^m-1=\prod_{i=1}^s g_i(x) \prod_{j=1}^t (h_j(x)h_j^{\ast}(x)).
\end{equation}
Here $g_i(x)$ is self-reciprocal for $1\leq i \leq s$ and $h_j(x)$ and $h_j^{\ast}(x)$ are reciprocal pairs for $1 \leq j \leq t$, where the reciprocal of a monic polynomial $f(x)$ with non-zero constant term is defined as 
$$f^{\ast}(x)=f(0)^{-1} x^{\deg f} f(x^{-1}).$$
 By the Chinese Remainder Theorem (CRT), $R_m^\ell$ decomposes as
$$R_m^\ell=\left(\bigoplus_{i=1}^s \mathbb{G}_i^{\ell}\right) \bigoplus \left(\bigoplus_{j=1}^t \left(\mathbb{H'}_j^\ell\bigoplus \mathbb{H''}_j^\ell\right) \right),$$
where for $1\leq i\leq s$, $\mathbb{G}_i=\mathbb{F}_q[x]/ \langle g_i(x) \rangle$, and for $1\leq j \leq t$, $\mathbb{H'}_j=\mathbb{F}_q[x] / \langle h_j(x) \rangle$ and $\mathbb{H''}_j=\mathbb{F}_q[x]/ \langle h_j^{\ast}(x)\rangle$. If $\xi$ is a primitive $m^{th}$ root of unity over $\F_q$  and $\xi^{u_i}$, $\xi^{v_j}$ and $\xi^{-v_j}$ are roots of $g_i(x)$, $h_j(x)$ and $h_j^{\ast}(x)$, respectively, then we have
$\mathbb{G}_i\cong \mathbb{F}_q(\xi^{u_i})$, $\mathbb{H'}_j\cong \mathbb{F}_q(\xi^{v_j})\cong \mathbb{F}_q(\xi^{-v_j}) \cong \mathbb{H{''}}_j$. Since the degree of a self-reciprocal polynomial is even, the degree of $\G_i$ over $\F_q$ is even for all $i$, except the components corresponding to the self-reciprocal irreducible factors $(x\pm 1)$ of $x^m-1$. 

Via the CRT decomposition of $R_m^\ell$, an $\ell$-QC code $C$ can be decomposed as 
\begin{equation}\label{CRTQC}
	C=\left(\bigoplus_{i=1}^s C_i \right)\bigoplus \left(\bigoplus_{j=1}^t \left(C'_j \bigoplus C''_j \right) \right),
\end{equation}
where $C_i, C'_j , C''_j$ are linear codes of length $\ell$ over the fields $\mathbb{G}_i , \mathbb{H'}_j , \mathbb{H''}_j$, respectively. These are called the constituents of $C$. It is known that the dual code $C^\bot$ is also $\ell$-QC code and it decomposes into constituents as
\begin{equation}\label{CRTDQC}
	C^{\perp}=\left(\bigoplus_{i=1}^s C_i^{\perp_h} \right) \bigoplus \left(\bigoplus_{j=1}^t \left({C''}_j^{\perp}\bigoplus {C'}_j^{\perp}\right) \right).
\end{equation}
We refer to \cite{AQC,GENERATOR} for (\ref{CRTQC}) and (\ref{CRTDQC}). Hence the hull dimension of the $\ell$-QC code $C$ over $\F_q$ is
\begin{equation} \label{QC hull dim}
h(C)=\sum_{i=1}^s \deg g_i(x) \; h_h(C_i)+\sum_{j=1}^t \deg h_j(x)[\dim({C'}_j \cap {C''}_j^{\perp}) + \dim({C''}_j\cap {C'}_j^{\perp})]. 
\end{equation} 
If $C$ and $D$ are $\ell$-QC codes of length $m\ell$ over $\F_q$ with constitutents $C_i,C'_j,C''_j$ and $D_i,D'_j,D'',j$, respectively (cf. (\ref{CRTQC})), then $(C,D)$ is LCP of codes if and only if 
\begin{equation}\label{LCPQC cond}
\mbox{$(C_i , D_i)$, $(C'_j, D'_j)$ and $(C{{''}}_j, D{{''}}_j)$ are LCP, for all $i,j$  (\cite[Theorem 3.1]{LCP}).} 
\end{equation}
As a consequence, $C$ is LCD if and only if
\begin{equation} \label{LCDQC cond}
\mbox{$C_i\cap C_i^{\perp_h}=\{0\}$, $C'_j \cap {C''}_j^{\perp}=\{0\}$ and $C''_j\cap {C'}_j^{\perp}=\{0\}$ for all $i,j$ (\cite[Theorem 3.1]{LCDQC}).}
\end{equation}

\section{One-Generator QC Codes}\label{1-gen sec}
We continue with the notation and assumptions in Section \ref{prelim}. In particular, we assume $\gcd(m,q)=1$. 

Let $C=\langle (a_1(x), \ldots, a_\ell(x)) \rangle \subset R_m^\ell$ be a 1-generator $\ell$-QC code. Constituents of $C$ can be described as follows (\cite[Equation 2.3]{LCDQC}):
\begin{eqnarray} \label{consts 1G}
C_i &= & Span_{\mathbb{G}_i}\lbrace (a_1(\xi^{u_i}), \ldots , a_\ell(\xi^{u_i})) \rbrace, \nonumber \\
C'_j&= & Span_{\mathbb{H}'_j}\lbrace (a_1(\xi^{v_j}), \ldots , a_\ell(\xi^{v_j})) \rbrace,\\
C{''}_j&= & Span_{\mathbb{H}{''}_j}\lbrace (a_1(\xi^{-v_j}), \ldots , a_\ell(\xi^{-v_j})) \rbrace. \nonumber
\end{eqnarray}
The generator polynomial of $C$ is defined by
$$g(x):= \gcd(a_1(x), \ldots, a_\ell(x), x^m-1).$$
The monic polynomial $h(x)$ of the least degree, which satisfies $h(x)a_i(x)=0$ for all $1\leq i \leq \ell$, is called the parity check polynomial of $C$. The polynomials $g(x)$ and $h(x)$ are unique, they satisfy the equation $g(x)h(x)=x^m-1$ in $\F_q[x]$ and   
$$\dim C=m-\deg g(x)=\deg h(x) \ \mbox{(cf. \cite{1GQC}}).$$ An $[m\ell,k]_q$ 1-generator $\ell$-QC code is called maximal if $k=m$.
For a maximal 1-generator QC code, we clearly have $g(x)=1$ and $h(x)=x^m-1.$

We start with describing the hull dimension of 1-generator QC codes. 

\begin{theorem}\label{1GHULL}
	Let $C=\langle (a_1(x), \ldots, a_\ell(x))\rangle$ be a 1-generator $\ell$-QC code of length $m\ell$ over $\mathbb{F}_q$, whose parity check polynomial is $h(x)$. Then the hull dimension of $C$ is $h(C)=\deg u(x)$, where 
	$$u(x)=\gcd\left(\sum_{r=1}^\ell a_r(x)a_r(x^{m-1}), h(x)\right).$$
\end{theorem}
\begin{proof}
By (\ref{consts 1G}), a constituent of $C$ is either 0 or 1 dimensional over its field of definition. We analyze each constituent's contribution to the hull in three cases. Recall that the polynomials $g_i(x),h_j(x),h^{\ast}_j(x)$ are irreducible factors of $x^m-1$ (cf. (\ref{factors})) and they correspond to the fields of definition of the constituents. On the other hand, the polynomials $g(x)$ and $h(x)$ stand for the generator and parity check polynomials of $C$, respectively.
 
\noindent \textbf{Case 1.} For any $i\in \{1,\ldots ,s\}$, $C_i\cap C_i^{\perp_h}\ne \lbrace 0 \rbrace$ if and only if  $C_i \ne \lbrace 0 \rbrace$ and $C_i\subseteq C_i^{\perp_h}$. Note that $C_i\not=\{0\}$ if and only if $g_i(x)\mid h(x)$. On the other hand, $C_i\subseteq C_i^{\perp_h}$ if and only if 
$$\sum_{r=1}^\ell a_r(\xi^{u_i})a_r(\xi^{-u_i})=0.$$ 
This is equivalent to the condition
$$g_i(x) \mid \sum_{r=1}^\ell a_r(x)a_r(x^{m-1}).$$

\noindent \textbf{Case 2.} For any $j\in \{1,\ldots ,t\}$, ${C'}_j \cap {C{''}}_j^{\perp} \ne \lbrace 0 \rbrace$ if and only if $C'_j\ne \lbrace 0 \rbrace$ and 
$C'_j \subseteq {C{''}}_j^{\perp}$. The first condition is equivalent to $h_j(x) \mid h(x)$, whereas the second condition amounts to 
$$\sum_{r=1}^\ell a_r(\xi^{v_j})a_r(\xi^{-v_j})=0.$$ 
This is equivalent to the condition
$$h_j(x) \mid \sum_{r=1}^\ell a_r(x)a_r(x^{m-1}).$$

\noindent \textbf{Case 3.} Arguing as in Case 2, we can see that ${C{''}}_j\cap {C'}_j^{\perp}\ne \lbrace 0 \rbrace$ if and only if $h_j^{\ast}(x)\mid h(x)$ and $h_j^{\ast}(x) \mid \sum_{r=1}^\ell a_r(x)a_r(x^{m-1})$.

Putting these together, we reach the result via (\ref{QC hull dim}). 
\end{proof}

The following is an immediate consequence of Theorem \ref{1GHULL}. We note that the LCD characterization for a special class of maximal 1-generator 2-QC codes (namely, double circulant codes) was given in \cite[ Theorem 5.1]{LCDQC}. Corollary \ref{1GENLCD} generalizes this result.
\begin{corollary}\label{1GENLCD}
Let $C=\langle (a_1(x), \ldots, a_\ell(x))\rangle$ be a 1-generator $\ell$-QC code of length $m\ell$ over $\mathbb{F}_q$, whose parity check polynomial is $h(x)$. Then $C$ is LCD if and only if 
\begin{align}\label{condition}
\gcd\left(\sum_{r=1}^\ell a_r(x)a_r(x^{m-1}), h(x)\right)=1.
\end{align}
\end{corollary}

Yang and Massey showed that a cyclic code is LCD if and only if its generator polynomial is self-reciprocal (\cite{LCDCYCLIC}). The next result shows that a self-reciprocal generator polynomial is a necessary condition for an LCD 1-generator QC code. However, it is not sufficient as shown in Remark \ref{sr not suff}.
\begin{proposition}\label{1GLCDSELF}
Let $C=\langle (a_1(x), \ldots, a_\ell(x))\rangle$ be a 1-generator $\ell$-QC code with the generator polynomial $g(x)$. If $C$ is LCD then $g(x)$ is self-reciprocal.
\end{proposition}
\begin{proof}
Suppose $g(x)$ is not self-reciprocal. Then there exists $h_j(x)$ such that $h_j(x)\mid g(x)$ but $h_j^{\ast}(x)\nmid g(x)$ (i.e., $h_j^{\ast}(x)\mid h(x)$). Since $h_j(x)\mid g(x)$, we have that $h_j(x)\mid a_r(x)$, for all $1\leq r \leq \ell$. Therefore $h_j^{\ast}(x) \mid a_r(x^{m-1})$ for all $r$. Hence  
$$h_j(x)\mid \gcd \left(\sum_{r=1}^\ell a_r(x)a_r(x^{m-1}), h(x)\right),$$
which contradicts the assumption that $C$ is LCD.
\end{proof}

\begin{remark} \label{sr not suff}
\cite[Theorem 4.4]{EHSQC} claims that a one-generator QC code is LCD if and only if the generator polynomial of the code is self-reciprocal and the condition \eqref{condition} is satisfied. Corollary \ref{1GENLCD} already shows that this claim is not true. Moreover the converse of Proposition \ref{1GLCDSELF} does not hold, as the following example shows. 

Let $C=\langle(x^2+x,x^2+1)\rangle$ be a binary 1-generator 2-QC code of length 6 (i.e.,  $m=3$).  Note that $g(x)=x+1$, $h(x)=x^2+x+1$ and hence $C$ is of dimension 2. The generator polynomial $g(x)$ is self-reciprocal but it is easy to see that $h(C)=2$ (cf. Theorem \ref{1GHULL}). Therefore $C$ is not LCD. 
\end{remark}

Tables 1 and 2 illustrate binary and ternary maximal 1-generator 2-QC  LCD codes $\langle (a_1(x),a_2(x) \rangle$ of length $2m$. The search is done by  Magma \cite{MAGMA} for random $a_1(x), a_2(x)$ in $R_m$ satisfying the condition in Theorem \ref{1GHULL}. In the tables $d$ presents the best minimum distance we obtained from this type of QC codes and $d^{\ast}$ represents optimal minimum distance of linear codes of length $2m$ and dimension $m$ (\cite{CDTB}).
\begin{table}[h!]
	\centering
	\begin{tabular}{||c c c c c||} 
		\hline
		$m$ &\qquad $d$ &\qquad $d^{\ast}$ &\qquad $a_1(x)$ &\qquad $a_2(x)$ \\ [0.5ex] 
		\hline\hline
		3 &\qquad  2 &\qquad 3 &\qquad $x+1$ &\qquad $x^2+x+1$\\ 
		5 &\qquad 3 &\qquad 4 &\qquad $x^3+1$ &\qquad $x^2+x+1$ \\
		7 &\qquad 4 &\qquad 4 &\qquad $x^2+1$ &\qquad $x^3+x+1$ \\
		9 &\qquad 5 &\qquad  6 &\qquad $x^5+x+1$ &\qquad $x^5+x^2+x+1$\\
		11 &\qquad 6 &\qquad 7 &\qquad $x^4+1$ &\qquad $x^8+x^7+x^6+x^2+1$ \\
        13 &\qquad 7 &\qquad 7 &\qquad $x^5+1$ &\qquad $x^{11}+x^9+x^6+x^3+1$ \\
		15 &\qquad 7 &\qquad 8 &\qquad $x^6+x^2+x+1$ &\qquad $x^5+x+1$ \\
		17 &\qquad 8 &\qquad 8 &\qquad $x^6+x^4+x+1$ &\qquad $x^5+x^4+x^3+x+1$ \\[1ex] 
		\hline
	\end{tabular}
	\caption{\label{tab:table-1} Binary maximal 1-generator 2-QC LCD Codes.}
\end{table}

\begin{table}[h!]
	\centering
	\begin{tabular}{||c c c c c||} 
		\hline
		$m$ &\qquad $d$ &\qquad $d^{\ast}$ &\qquad $a_1(x)$ &\qquad $a_2(x)$ \\ [0.5ex] 
		\hline\hline
		4 &\qquad  4 &\qquad 4 &\qquad $x+1$ &\qquad $x+2$\\ 
		5 &\qquad 4 &\qquad 5 &\qquad $x+2$ &\qquad $2x+2$ \\
		7 &\qquad 6 &\qquad 6 &\qquad $2x^6+2x^4+2x^3+2x+1$ &\qquad $2x^4+x+2$ \\
		8 &\qquad 6 &\qquad  6 &\qquad $x^3+2x+2$ &\qquad $x^2+2x+2$\\
		10 &\qquad 6 &\qquad 7 &\qquad $x^3+x+1$ &\qquad $x^2+2x+1$ \\ 
		11 &\qquad 7 &\qquad 8 &\qquad $x^3+2x+2$ &\qquad $x^3+2x^2+2x+1$ \\
		13 &\qquad 7 &\qquad 8 &\qquad $x^3+x^2+x+1$ &\qquad $x^4+x^2+2x+2$ \\
		14 &\qquad 8 &\qquad 9 &\qquad $x^4+x^2+x+2$ &\qquad $x^3+2x^2+x+1$ \\[1ex] 
		\hline
	\end{tabular}
	\caption{\label{tab:table-2} Ternary maximal 1-generator 2-QC LCD Codes.}
\end{table}

Next, we study LCP of 1-generator QC codes. The following simple observation shows that LCP of 1-generator QC codes are rather constrained. 

\begin{lemma}\label{1G constr}
If $(C,D)$ is LCP of 1-generator $\ell$-QC codes of length $m\ell$, then $\ell=2$ and both $C$ and $D$ are maximal.
\end{lemma}

\begin{proof}
Since the pair is linear complementary, we have $\dim(C)+\dim(D)=m\ell$. However, the maximal dimension for 1-generator QC codes is $m$. Therefore, $\ell=2$ and the dimension of each code is $m$. 
\end{proof}

The following result characterizes LCP of maximal 1-generator 2-QC codes. We note that this result was proved for double circulant codes in \cite[Proposition 3.2]{LCP}. 

\begin{proposition}\label{1GLCP}
Let $C=\langle(a_1(x),a_2(x))\rangle$ and $D=\langle(b_1(x), b_2(x))\rangle$ be maximal 1-generator 2-QC codes. Then $(C,D)$ is LCP of codes if and only if 
$$\gcd(a_1(x)b_2(x)-a_2(x)b_1(x), x^m-1)=1.$$
\end{proposition}

\begin{proof}
Constituents of both codes, which are described in (\ref{consts 1G}), lie in 2-dimensional spaces over certain extensions of $\F_q$. Therefore by (\ref{LCPQC cond}), $(C,D)$ is LCP of codes if and only if  the following $2\times 2$ matrices are of full rank for all $i,j$: 
$$\left(\begin{array}{cc} a_1(\xi^{u_i}) & a_2(\xi^{u_i})\\ b_1(\xi^{u_i}) & b_2(\xi^{u_i})\end{array}\right) \ \ \left(\begin{array}{cc} a_1(\xi^{v_j}) & a_2(\xi^{v_j})\\ b_1(\xi^{v_j}) & b_2(\xi^{v_j})\end{array}\right) \ \ \left(\begin{array}{cc} a_1(\xi^{-v_j}) & a_2(\xi^{-v_j})\\ b_1(\xi^{-v_j}) & b_2(\xi^{-v_j})\end{array}\right)$$
This is true if and only if no irreducible factor of $x^m-1$ divides the polynomial $a_1(x)b_2(x)-a_2(x)b_1(x)$.
\end{proof}

A 1-generator QC code of the form $C=\langle (1,a(x))\rangle \subset R_m^2$ is called a double circulant (DC) code. A DC code is a maximal 1-generator QC code. We provide ternary LCP of DC codes with good security parameters in Table 3. Here, $d$ represents the security parameter of the pair and $d^{\ast}$ is the best minimum distance for ternary $[2m,m]$ linear codes (\cite{CDTB}).

\begin{table}[h!]\label{TTLCPDC}
	\centering
	\small
	\begin{tabular}{||c c c c c||} 
		\hline
		$m$ & \qquad $d$ &\qquad $d^{\ast}$ &\qquad $a(x)$ & \qquad $b(x)=-a(x^{m-1})$\\ [0.5ex] 
		\hline\hline
		4 &\qquad  4 &\qquad 4 &\qquad $x^3+2x+1$ &\qquad $x^3+2x+2$ \\ 
		5 &\qquad 4 &\qquad 5 &\qquad $x^4+x+2$  &\qquad $x^4+2x+1$\\
		7 &\qquad 5 &\qquad 6 &\qquad $x^6+x^3+x+1$ &\qquad $2x^6+2x^4+2x+2$ \\
		8 &\qquad 6 &\qquad  6 &\qquad $x^7+x^3+x^2+2x+2$ &\qquad $x^7+2x^6+2x^5+2x+1$ \\
		10 &\qquad 7 &\qquad 7 &\qquad $x^9+x^5+x^4+x^2+x+2$ &\qquad $2x^9+2x^8+2x^6+2x^5+2x+1$ \\ 
		11 &\qquad 7 &\qquad 8 &\qquad $2x^{10}+2x^9+2x^8+x^5+x^2+2$ &\qquad $2x^9+2x^6+x^3+x^2+x+1$ \\[1ex] 
		\hline
	\end{tabular}
	\caption{\label{tab:table-8} Ternary LCP of DC Codes.}
\end{table}

We conclude this section with further observations on DC codes. A formula for the hull dimension of DC codes follows from Theorem \ref{1GHULL}, which is 
$$h(C)=\deg \gcd(1+a(x)a(x^{m-1}), x^m-1).$$
The next result characterizes the existence of DC codes of hull dimension 1, which is of interest for various applications (cf. Introduction), and it also provides a necessary condition for the existence of DC codes of odd hull dimension. The proof explicitly describes the DC codes of hull dimension 1 as well.

\begin{theorem}\label{1dIMHULLDC}
i. There exists a DC code of hull dimension one over $\F_q$ if and only if $q\equiv 1$ (mod 4) or $q$ is even.

ii. If there exists a DC code with odd hull dimension over $\F_q$, then $q\equiv 1$ (mod 4) or $q$ is even.
\end{theorem}

\begin{proof}
i. Let $C=\langle (1,a(x))\rangle \subset R_m^2$. The only linear factors of $x^m-1$ are $x-1$ and $x+1$, where the latter can only occur in the case of $\F_q$ with odd characteristics and $m$ is even. Therefore hull dimension 1 is only possible if there is a contribution to the hull only from one of these linear irreducible factors (cf. (\ref{QC hull dim})). Let us denote the constituent corresponding to $x-1$ by $C_1\subset \F_q^2$, which is a 1-dimensional space spanned by $(1,a(1))$. It is easy to observe that $C_1^\bot =Span_{\F_q}\{(-a(1),1)\}$. Therefore $h(C_1)=1$ if and only if $a(1)^2=-1$. This implies that $q\equiv 1$ (mod 4) or $q$ is even. 

For the converse, let us construct $a(x)$ so that the resulting DC code has a 1-dimensional hull. In the case  $q\equiv 1$ (mod 4), let $\alpha \in \mathbb{F}_q\setminus \{0\}$ such that $\alpha^2=-1$ and set $a(x)=x-(\alpha+1)$. Then,
\begin{eqnarray*}
1+a(x)a(x^{m-1})&=&1+ (x-(\alpha+1))(x^{m-1}-(\alpha+1))\\
&=&(\alpha +1)(2-x-x^{m-1}).
\end{eqnarray*}
For an $m^{th}$ root of unity $\zeta$ to be a root of this polynomial, we have  
$$\zeta^{-1}+\zeta-2=0 \; \iff \; \zeta^2-2\zeta+1=(\zeta-1)^2=0 \; \iff \; \zeta=1.$$
Hence $\gcd \left( 1+a(x)a(x^{m-1}), x^m-1 \right)=x-1$ and $h(C)=1$ by Theorem \ref{1GHULL}. For $q$ even, let $u(x)=(x^m-1)/(x-1)$ and let $\beta :=u(1)\ne 0$. If we set $a(x)=u(x)+(\beta+1)$ and $v(x)=1+a(x)a(x^{m-1})$, we have
$$v(1)=1+(u(1)+\beta+1)(u(1)+\beta+1)=1+1=0.$$ 
On the other hand, if $\zeta\ne 1$ is another $m^{th}$ root of unity, we have
$$v(\zeta)=1+(u(\zeta)+\beta+1)(u(\zeta^{-1})+\beta+1).$$
Since $u(\zeta)=u(\zeta^{-1})=0$, we obtain
$$v(\zeta)=1+(\beta+1)^2=1+\beta^2+1=\beta^2\not=0.$$
Therefore, $\gcd(1+a(x)a(x^{m-1}), x^m-1)=x-1$ and $h(C)=1$ again. 

ii. All self-reciprocal irreducible factors $g_i(x)$ of $x^m-1$ (cf. (\ref{factors})) other than $(x\mp 1)$ are of even degree. Therefore contribution to the hull dimension from the constituents corresponding to such $g_i(x)$ is even (cf. (\ref{QC hull dim})). For a pair of reciprocal irreducible factors $h_j(x),h_j^{\ast}(x)$, the corresponding constituents of $C$ are
$$C'_j=Span_{\mathbb{H'}}\{(1,a(\xi^{v_j}))\} \ \ \mbox{and} \ \ C''_j=Span_{\mathbb{H''}}\{(1,a(\xi^{-v_j}))\}.$$
Duals of these constituents are easily seen to be
$${C'}_j^{\perp}=Span_{\mathbb{H'}}\{(-a(\xi^{v_j}),1)\} \ \ \mbox{and} \ \ {C''}_j^{\perp}=Span_{\mathbb{H''}}\{(-a(\xi^{-v_j}),1)\}.$$
Therefore, $C'_j \cap 	{C''}_j^{\perp}\not=\{0\}$ if and only if $a(\xi^{v_j})a(\xi^{-v_j})=-1$. On the other hand, ${C'}_j^{\perp} \cap 	{C''}_j\not=\{0\}$ if and only if $a(\xi^{v_j})a(\xi^{-v_j})=-1$. Hence, these two intersections are either simultaneously 0 or they are both of dimension 1. Hence, the contribution of such a pair of constituents to $h(C)$ is also even (cf. (\ref{QC hull dim})). Therefore an odd hull dimension can only be attained from a constituent corresponding to a linear factor of $x^m-1$, which implies the conditions on $q$ as in part i. 
\end{proof}

Tables 4 and 5 present the best possible minimum distances for binary and quinary DC codes with hull dimension 1. Here $d^{*}$ is the best known minimum distance for linear codes of length $2m$ and dimension $m$ (\cite{CDTB}), whereas $d$ is the best possible minimum distance which can be obtained from a DC code $C=\langle (1,a(x))\rangle$ of hull dimension 1. 
\begin{table}[h!]\label{BIN1HULLDC}
	\centering
	\begin{tabular}{||c c c c||} 
		\hline
		$m$ & \qquad $d$ &\qquad $d^{\ast}$ &\qquad $a(x)$ \\ [0.5ex] 
		\hline
		3 &\qquad  2 &\qquad 3 &\qquad $x^2+x+1$ \\ 
		5 &\qquad 4 &\qquad 4 &\qquad $x^4+x^2+1$ \\
		7 &\qquad 4 &\qquad 4 &\qquad $x^6+x^3+1$ \\
		9 &\qquad 6 &\qquad  6 &\qquad $x^8+x^7+x^5+x^3+x^2$ \\
		11 &\qquad 6 &\qquad 7 &\qquad $x^{10}+x^8+x^5+x^3+1$ \\ 
		13 &\qquad 6 &\qquad 7 &\qquad $x^{12}+x^4+x^3+x+1$ \\
		15 &\qquad 8 &\qquad 8 &\qquad $x^{14}+\cdots+x^7+x^4+x^3+x$ \\
		17 &\qquad 8 &\qquad 8 &\qquad $x^{16}+\cdots+x^{11}+x^5+x^3+x+1$ \\[1ex] 
		\hline
	\end{tabular}
	\caption{\label{tab:table-7} Binary DC Codes with hull dimension 1}
\end{table}

\begin{table}[h!]\label{QUI1HULLDC}
	\centering
	\begin{tabular}{||c c c c||} 
		\hline
		$m$ & \qquad $d$ &\qquad $d^{\ast}$ &\qquad $a(x)$ \\ [0.5ex] 
		\hline\hline
		3 &\qquad  3 &\qquad 4 &\qquad $3x^2+3x+1$ \\ 
		4 &\qquad 4 &\qquad 4 &\qquad $x^3+x^2+3x+3$ \\
		6 &\qquad 6 &\qquad 6 &\qquad $x^5+x^3+2x^2+2x+1$ \\
		7 &\qquad 6 &\qquad  6 &\qquad $x^4+x^3+x^2+2x+3$ \\
		8 &\qquad 7  &\qquad 7 &\qquad $x^5+2x^4+4x^3+2x^2+2x+2$ \\ 
		9 &\qquad 7 &\qquad 7 &\qquad $x^5+x^4+x^3+2x^2+x+2$ \\
		11 &\qquad 8 &\qquad 8 &\qquad $x^6+x^5+x^4+2x^3+x^2+4x+2$ \\
		12 &\qquad 8 &\qquad 8 &\qquad $x^7+x^6+4x^5+2x^4+4x^3+4x^2+3x+4$ \\[1ex] 
		\hline
	\end{tabular}
	\caption{\label{tab:table-7} Quinary DC Codes with hull dimension 1}
\end{table}
\newpage

The next example shows that the converse of part ii in Theorem \ref{1dIMHULLDC} is not always true. 

\begin{example}
Let $q=4$ and $m=9$. Then
$$x^9-1=(x+1)(x+\alpha)(x+\alpha^2)(x^3+\alpha)(x^3+\alpha^2),$$
where $\alpha$ is a primitive element of $\mathbb{F}_4$. Note that $(x+\alpha),(x+\alpha^2)$ and $(x^3+\alpha),(x^3+\alpha^2)$ are reciprocal pairs. For the following choices of $a(x)$, we have even hull dimension for the DC code $\langle (1,a(x)) \rangle \subset (\F_4[x]/\langle x^9-1 \rangle)^2$:

\begin{itemize}
\item $a(x)=\alpha^2x^8+\alpha^2x^7+\alpha^2x^6+x^3+x+1$: $[18,9,7]_4$ DC code with hull dimension 2.
\item $a(x)=x^7+x^6+x^5+\alpha x^4+\alpha^2x^3+\alpha^2x^2+\alpha x+\alpha^2$: $[18,9,7]_4$ DC code with hull dimension 6.
\end{itemize}

Let $q=5$ and $m=8$. Then
$$x^8-1=(x+1)(x+2)(x+3)(x+4)(x^2+2)(x^2+3).$$
Note that $(x+1),(x+4)$ are self-reciprocal and $(x+2),(x+3)$ and $(x^2+2),(x^2+3)$ are reciprocal to each other.  For the following choices of $a(x)$, we have even hull dimension for the DC code $\langle (1,a(x)) \rangle \subset (\F_5[x]/\langle x^8-1 \rangle)^2$:
	\begin{itemize}
		\item $a(x)=4x^7+4x^6+x^3+4x^2+3x+2$: $[16,8,7]_5$ DC code with hull dimension 2.
		\item $a(x)=4x^7+4x^6+4x^5+2x^3+4$: $[16,8,6]_5$ DC code with hull dimension 4.
	\end{itemize}
\end{example}

\section{Four-Circulant Codes}\label{4-circ sec}
We now investigate a class of 2-generator 4-QC codes. The code
$$C=\langle (1,0,a_1(x), a_2(x)),(0,1,-a_2(x^{m-1}),a_1(x^{m-1})) \rangle \subset R_m^4$$
is called a four-circulant (FC) code. By \cite[Equation 2.3]{LCDQC}, the following matrices generate the 2-dimensional constituents $C_i, C'_j,C^{''}_j$ of $C$ ($1\leq i\leq s$ and $1\leq j \leq t$):
\[\begin{array}{cc}
G_i=\begin{pmatrix}
	1 & 0 & a_1(\xi^{u_i}) & a_2(\xi^{u_i})\\
	0 & 1 & -a_2(\xi^{-u_i}) & a_1(\xi^{-u_i}) 
\end{pmatrix}  & 
G'_j=\begin{pmatrix}
	1 & 0 & a_1(\xi^{v_j}) & a_2(\xi^{v_j})\\
	0 & 1 & -a_2(\xi^{-v_j}) & a_1(\xi^{-v_j}) 
\end{pmatrix}
\end{array}\]
\begin{equation}\label{consts-2}
G{''}_j=\begin{pmatrix}
	1 & 0 & a_1(\xi^{-v_j}) & a_2(\xi^{-v_j})\\
	0 & 1 & -a_2(\xi^{v_j}) & a_1(\xi^{v_j}) 
\end{pmatrix}
\end{equation}

We first describe the hull dimension of FC codes. 

\begin{theorem}\label{4CHULL}
Let $C=\langle (1,0,a_1(x), a_2(x)),(0,1,-a_2(x^{m-1}),a_1(x^{m-1})) \rangle \subset R_m^4$ be a FC code. Then the hull dimension of $C$ is $h(C)=2\deg u(x)$, where 
$$u(x)=\gcd \bigl(1+a_1(x)a_1(x^{m-1})+a_2(x)a_2(x^{m-1}), x^m-1\bigr).$$
In particular, an FC code of odd hull dimension does not exist over any finite field.
\end{theorem}

\begin{proof}
By (\ref{herm hull dim}) and \cite[Theorem 2.1]{LINT}, which generalizes (\ref{hull dim}), dimensions that contribute to the hull dimension $h(C)$ of $C$ are the following (cf. (\ref{QC hull dim})):
\begin{equation} \label{4C hull contr}
h_h(C_i)=2-\hbox{rank}(G_i\bar{G_i}^T),\ \
\dim(C'_j\cap C''^{\bot}_j)=2-\hbox{rank}(G'_jG''^T_j)=2-\hbox{rank}(G''_jG'^T_j)= \dim(C''_j\cap C'^{\bot}_j). 
\end{equation}
Let 
$$A(x):=1+a_1(x)a_1(x^{m-1})+a_2(x)a_2(x^{m-1}).$$
Note that
$$
G_i\bar{G}_i^T=
\left(\begin{array}{cccc}
1 & 0 & a_1(\xi^{u_i}) & a_2(\xi^{u_i}) \\
0 & 1 & -a_2(\xi^{-u_i}) & a_1(\xi^{-u_i}) 
\end{array}\right) \left(\begin{array}{cc}
1 & 0 \\
0 & 1 \\
a_1(\xi^{-u_i}) & -a_2(\xi^{u_i})\\
a_2(\xi^{-u_i}) & a_1(\xi^{u_i})
\end{array}\right)=\left(\begin{array}{cc} A(\xi^{u_i}) & 0 \\ 0 & A(\xi^{u_i}) \end{array} \right).
$$
On the other hand, 
$$
G'_j G{''}_j^T=
\left(\begin{array}{cccc}
1 & 0 & a_1(\xi^{v_j}) & a_2(\xi^{v_j})\\
0 & 1 & -a_2(\xi^{-v_j}) & a_1(\xi^{-v_j})  
\end{array}\right) \left(\begin{array}{cc}
1 & 0 \\
0 & 1 \\
a_1(\xi^{-v_j}) & -a_2(\xi^{v_j})\\
a_2(\xi^{-v_j}) & a_1(\xi^{v_j})
\end{array}\right)=\left(\begin{array}{cc} A(\xi^{v_j}) & 0 \\ 0 & A(\xi^{v_j}) \end{array} \right).
$$
Hence we have (cf. (\ref{factors})),
\[\begin{array}{ll} \hbox{rank}(G_i\bar{G}_i^T)=\left\{ \begin{array}{ll} 0 & \mbox{if $g_i(x)|A(x)$}\\ 2 & \mbox{otherwise}\end{array} \right. , & \hbox{rank}(G'_j G{''}_j^T)=\hbox{rank}(G''_jG'^T_j)=\left\{\begin{array}{ll} 0 & \mbox{if $h_j(x)|A(x)$}\\ 2 & \mbox{otherwise}\end{array}  \right.  . \end{array}\]
Irreducible factors $h_j(x)$ and $h_j^{\ast}(x)$ of $x^m-1$ both divide $A(x)$ or neither does, since $A(x)$ is self-reciprocal. Combining these with (\ref{QC hull dim}) and (\ref{4C hull contr}), the result follows. 
\end{proof}

An immediate consequence of Theorem \ref{4CHULL} is the characterization of LCD FC codes. Although this analysis is carried out in \cite[Section 3]{4CLCD} and used, it is not explicitly stated. So, we state this result.

\begin{corollary}
Let $C=\langle (1,0,a_1(x), a_2(x)),(0,1,-a_2(x^{m-1}),a_1(x^{m-1})) \rangle \subset R_m^4$ be a FC code. Then $C$ is LCD if and only if 
$$\gcd(1+a_1(x)a_1(x^{m-1})+a_2(x)a_2(x^{m-1}), x^m-1)=1.$$
\end{corollary}

Tables 6 and 7 present the best possible distances of binary and ternary LCD FC codes. Meanings of $d$ and $d^{\ast}$ are analogous to previous tables. 

\begin{table}[h!]\label{BINFCLCD}
	\centering
	\begin{tabular}{||c c c c c||} 
		\hline
		$m$ & \qquad $d$ &\qquad $d^{\ast}$ &\qquad $a_1(x)$ &\qquad $a_2(x)$  \\ [0.5ex] 
		\hline\hline
		3 &\qquad  2 &\qquad 4 &\qquad $x+1$ &\qquad $x^2+x$ \\ 
		5 &\qquad 5 &\qquad 6 &\qquad $x^2$ &\qquad $x^2+x+1$ \\
		7 &\qquad 6 &\qquad 8 &\qquad $x^6+x^5+x^4+x^3$ &\qquad $x+1$ \\
		9 &\qquad 6 &\qquad  8 &\qquad $x^7+x^6+x^5+x^3+x$ &\qquad $x^3+x+1$ \\
		11 &\qquad 9 &\qquad 10 &\qquad $x^5+x^3+x^2$ &\qquad $x^7+x^6+x^5+x+1$ \\ 
		13 &\qquad 8 &\qquad 10 &\qquad $x^7+x^6+x+1$ &\qquad $x^4+x^3+x^2+1$ \\[1ex] 
		\hline
	\end{tabular}
	\caption{\label{tab:table-9} Binary LCD FC codes.}
\end{table}
\begin{table}[h!]\label{TERFCLCD}
	\centering
	\begin{tabular}{||c c c c c||} 
		\hline
		$m$ & \qquad $d$ &\qquad $d^{\ast}$ &\qquad $a_1(x)$ &\qquad $a_2(x)$  \\ [0.5ex] 
		\hline\hline
		4 &\qquad  6 &\qquad 6 &\qquad $2x^3+x^2+1$ &\qquad $2x^3+1$ \\ 
		5 &\qquad 7 &\qquad 7 &\qquad $x^4+2x^2+x+2$ &\qquad $2x^4+2x^2+1$ \\
		7 &\qquad 8 &\qquad 9 &\qquad $x^6+2x^5+x^3+x$ &\qquad $2x^5+x^4+x^3+2$ \\
		8 &\qquad 9 &\qquad  10 &\qquad $2x^5+x^2+1$ &\qquad $x^5+x^4+x^3+2x+1$ \\[1ex] 
		\hline
	\end{tabular}
	\caption{\label{tab:table-10} Ternary LCD FC codes.}
\end{table}

We finish by characterizing LCP of FC codes. 

\begin{theorem}\label{FCLCP}
Let 
\begin{eqnarray*}
C&=&\langle (1,0,a_1(x), a_2(x)),(0,1,-a_2(x^{m-1}),a_1(x^{m-1})\rangle\\
D&=&\langle(1,0,b_1(x), b_2(x)),(0,1,-b_2(x^{m-1}),b_1(x^{m-1})\rangle
\end{eqnarray*} 
be FC codes of length $4m$ over $\mathbb{F}_q$. Then, $(C,D)$ is LCP if and only if 
$$\gcd \left( \sum_{t=1}^2 [(a_t(x)-b_t(x))(a_t(x^{m-1})-b_t(x^{m-1})] , x^m-1 \right)=1.$$
\end{theorem}

\begin{proof}
We will denote the $2\times 4$ matrices that generate the constituents of $C$ by $G_{C_i}, G_{C'_j}, G_{C''_j}$. Likewise, we denote the corresponding matrices for $D$ by $G_{D_i}, G_{D'_j}, G_{D''_j}$. The forms of these matrices are clear from the previous analysis and will also be evident in what follows. By (\ref{LCPQC cond}), $(C,D)$ is LCP of codes if and only if the following matrices are of full rank for all $i,j$: 
$$\begin{pmatrix}
1 & 0 & a_1(\xi^{u_i}) & a_2(\xi^{u_i})\\
0 & 1 & -a_2(\xi^{-u_i})& a_1(\xi^{-u_i}) \\
1 & 0 & b_1(\xi^{u_i}) & b_2(\xi^{u_i})\\
0 & 1 & -b_2(\xi^{-u_i})& b_1(\xi^{-u_i}) 
\end{pmatrix} \hspace{0.1cm}  \begin{pmatrix}
1 & 0 & a_1(\xi^{v_j}) & a_2(\xi^{v_j})\\
0 & 1 & -a_2(\xi^{-v_j})& a_1(\xi^{-v_j})\\
1 & 0 & b_1(\xi^{v_j}) & b_2(\xi^{v_j})\\
0 & 1 & -b_2(\xi^{-v_j})& b_1(\xi^{-v_j}) 
\end{pmatrix} \hspace{0.1cm} \begin{pmatrix}
1 & 0 & a_1(\xi^{-v_j}) & a_2(\xi^{-v_j})\\
0 & 1 & -a_2(\xi^{v_j})& a_1(\xi^{v_j})\\
1 & 0 & b_1(\xi^{-v_j}) & b_2(\xi^{-v_j})\\
0 & 1 & -b_2(\xi^{v_j})& b_1(\xi^{v_j}) 
\end{pmatrix}$$
By elementary row operations, the first matrix turns into 
$$\begin{pmatrix}
1 & 0 & a_1(\xi^{u_i}) & a_2(\xi^{u_i})\\
0 & 1 & -a_2(\xi^{-u_i})& a_1(\xi^{-u_i}) \\
0 & 0 & b_1(\xi^{u_i}) -a_1(\xi^{u_i})  & b_2(\xi^{u_i})-a_2(\xi^{u_i})\\
0 & 0 & -b_2(\xi^{-u_i})+a_2(\xi^{-u_i})& b_1(\xi^{-u_i})-a_1(\xi^{-u_i}) 
\end{pmatrix}, $$
which is of rank 4 if and only if the $2\times 2$ minor in the lower right corner is nonzero. This is equivalent to demanding that the irreducible factor $g_i(x)$ of $x^m-1$ does not divide the polynomial
$$ [(a_1(x)-b_1(x))(a_1(x^{m-1})-b_1(x^{m-1})] +  [(a_2(x)-b_2(x))(a_2(x^{m-1})-b_2(x^{m-1})].$$ 
The same analysis applied to the second and third matrices yield analogous consequences for the irreducible factors $h_j(x)$ and $h^{\ast}_j(x)$ of $x^m-1$. Hence the result follows. 
\end{proof}

Table 8 presents ternary LCP of FC codes with good security parameters. Here, $d$ represents the best security parameter via exhaustive search for the corresponding $m$ value and $d^{\ast}$ is the best minimum distance for a $[4m,2m]$ ternary linear code. 
\begin{table}[h!]\label{FCLCPEQ}
	\centering
	\begin{tabular}{||c c c c c||} 
		\hline
		$m$ & \qquad $d$ &\qquad $d^{\ast}$ &\qquad $a_1(x)$ &\qquad $a_2(x)$  \\ [0.5ex] 
		\hline\hline
		4 &\qquad  6 &\qquad 6 &\qquad $2x^3+2x^2+x$ &\qquad $x^2+1$ \\ 
		5 &\qquad 7 &\qquad 7 &\qquad $x^2+2x+1$ &\qquad $2x^3+x+1$ \\
		7 &\qquad 9 &\qquad 9 &\qquad $x^3+2x^2+1$ &\qquad $2x^5+2x^3+2x^2+x+1$ \\[1ex] 
		\hline
	\end{tabular}
	\caption{\label{tab:table-11} Ternary LCP of FC codes.}
\end{table}

\section*{Acknowledgment}
This paper is dedicated to Professor Sudhir Ghorpade on the occasion of his sixtieth birthday. T. Kalayc\i \ is supported by T\"UB\.ITAK Project under Grant 120F309.

\end{document}